\documentclass{article}

\usepackage[english]{babel}

\usepackage[letterpaper,top=2cm,bottom=2cm,left=3cm,right=3cm,marginparwidth=1.75cm]{geometry}

\usepackage{amsmath}
\usepackage{graphicx}
\usepackage[colorlinks=true, allcolors=blue]{hyperref}
\usepackage[OT1]{fontenc} 
\usepackage{microtype}
\usepackage{graphicx}
\usepackage{subfigure}
\usepackage{booktabs} 

\usepackage{hyperref}

\usepackage{tablefootnote}

\usepackage{amsmath}
\usepackage{amssymb}
\usepackage{mathtools}
\usepackage{amsthm}

\usepackage{indentfirst} 
\usepackage{booktabs}
\usepackage{amsfonts}
\usepackage{amsmath}
\usepackage{multirow}
\usepackage{amsthm}
\usepackage{diagbox}
\usepackage{bm}
\usepackage{color}
\usepackage{graphicx}
\usepackage{natbib}
\usepackage{enumitem}
\usepackage{algorithm}
\usepackage{algorithmic}
\usepackage{mathtools}
\usepackage{tablefootnote}
\usepackage{thmtools} 
\usepackage{thm-restate}
\usepackage{xcolor}         
\definecolor{pku-red}{RGB}{139,0,18}
\usepackage[capitalize,noabbrev,nameinlink]{cleveref}
\usepackage[textsize=tiny]{todonotes}
\usepackage{bbm}
\usepackage[normalem]{ulem}
\usepackage{comment}

\theoremstyle{plain}
\newtheorem{theorem}{Theorem}[section]
\newtheorem{proposition}[theorem]{Proposition}
\newtheorem{lemma}[theorem]{Lemma}

\theoremstyle{definition}

\theoremstyle{remark}
\newtheorem{remark}[theorem]{Remark}

\newcommand{\EE}{\mathbb{E}}

\newcommand{\RR}{\mathbb{R}}

\newcommand{\Aa}{\mathcal{A}}
\newcommand{\Bb}{\mathcal{B}}

\newcommand{\Ss}{\mathcal{S}}

\newcommand{\bmb}{{\bm{b}}}

\newcommand{\bmu}{{\bm{u}}}
\newcommand{\bmv}{{\bm{v}}}
\newcommand{\bmw}{{\bm{w}}}
\newcommand{\bmx}{{\bm{x}}}

\DeclareMathOperator*{\argmax}{argmax}

\newcommand{\rev}{{R}}
\newcommand{\name}{\text{ZFO-VVCA}}
\newcommand{\MAW}{\mathrm{MAW}}
\newcommand{\AB}{\mathrm{AB}}

\newcommand{\bw}{{\bm{w}}}
\newcommand{\bl}{{\bm\lambda}}
\newcommand{\ba}{{\bm\alpha}}

\newcommand{\bnoise}{{\bmu}}
\newcommand{\noise}{{u}}

\newcommand{\email}[1]{\texttt{#1}}

\title{\bfseries{Automated Deterministic Auction Design with Objective Decomposition}}

\author{
\textbf{Zhijian Duan}$^{1,*}$, 
\textbf{Haoran Sun}$^{1,*}$,  
\textbf{Yichong Xia}$^{1,\#}$, 
\textbf{Siqiang Wang}$^{1,\#}$,\\ 
\textbf{Zhilin Zhang}$^2$, 
\textbf{Chuan Yu}$^2$, 
\textbf{Jian Xu}$^2$, 
\textbf{Bo Zheng}$^2$, 
\textbf{Xiaotie Deng}$^1$
\\
$^{1}$Peking University
$^{2}$Alibaba Group
\\
\email{zjduan@pku.edu.cn},\email{\{sunhaoran0301,xiayc,smswangsiqiang\}@stu.pku.edu.cn},\\
\email{\{zhangzhilin.pt,yuchuan.yc,xiyu.xj,bozheng\}@alibaba-inc.com},\\
\email{xiaotie@pku.edu.cn}
}
\date{}

\begin{document}
	\maketitle
	\def\thefootnote{*}\footnotetext{Equal contribution.}
	\def\thefootnote{\#}\footnotetext{Equal contribution.}

	\begin{abstract}
Identifying high-revenue mechanisms that are both dominant strategy incentive compatible (DSIC) and individually rational (IR) is a fundamental challenge in auction design.
While theoretical approaches have encountered bottlenecks in multi-item auctions, there has been much empirical progress in automated designing such mechanisms using machine learning. 
However, existing research primarily focuses on randomized auctions, with less attention given to the more practical deterministic auctions.
Therefore, this paper investigates the automated design of deterministic auctions and introduces OD-VVCA, an objective decomposition approach for automated designing Virtual Valuations Combinatorial Auctions (VVCAs).
Firstly, we restrict our mechanism to deterministic VVCAs, which are inherently DSIC and IR.
Afterward, we utilize a parallelizable dynamic programming algorithm to compute the allocation and revenue outcomes of a VVCA efficiently.
We then decompose the revenue objective function into continuous and piecewise constant discontinuous components, optimizing each using distinct methods.
Extensive experiments show that OD-VVCA achieves high revenue in multi-item auctions, especially in large-scale settings where it outperforms both randomized and deterministic baselines, indicating its efficacy and scalability.
\end{abstract}
	
	\section{Introduction}
Finding high-revenue mechanisms that are both \emph{dominant strategy incentive compatible} (DSIC) and \emph{individually rational} (IR) is a fundamental problem in auction design, where DSIC and IR mechanisms incentivize bidders to truthfully report their valuations and ensure they receive non-negative utilities by doing so.
While \citet{myerson1981optimal} characterized the optimal design for single-item auctions, theoretical approaches have faced significant challenges in general multi-item cases~\citep{dutting2023optimal}.
Consequently, there has been much progress in automated auction design, which formulates a constrained optimization problem and empirically seeks high-revenue mechanisms using machine learning~\citep{VVCA2015, dutting2019optimal}.

Existing research on the automated design of general multi-item auctions primarily focuses on randomized auctions, which produce randomized allocation outcomes~\citep{dutting2019optimal,curry2020certifying,peri2021preferencenet,rahme2021auction,rahme2021permutation,duan2022context,ivanov2022optimal,curry2022differentiable,duan2023scalable}, with deterministic auctions receiving less attention.
However, deterministic auctions offer practical advantages over randomized ones in real-world applications, because deterministic allocations are more interpretable and easier to implement, leading to widespread adoption in online advertising~\citep{edelman2007internet, liu2021neural, li2024deep}. 
Motivated by these practical considerations, this paper investigates the automated design of deterministic mechanisms for general multi-item auctions.

The automated design of deterministic mechanisms for general multi-item auctions faces two main challenges. 
Firstly, \emph{scalability} is a significant concern. 
This is because presently, the most effective approach to ensuring DSIC and IR in automated deterministic auction design is to constrain the mechanism to Affine Maximizer Auctions (AMAs)~\citep{AMA}. 
AMAs generalize Vickrey-Clarke-Groves (VCG) auctions~\citep{VCG} by assigning weights to each bidder and boosting each candidate allocation to achieve higher revenue than VCG while maintaining DSIC and IR. 
However, AMAs boost each candidate allocation, resulting in $O((n+1)^m)$ boost parameters for $n$ bidders and $m$ items, which is exponentially large. 
Even worse, computing the allocation result of an AMA, which requires finding the winning allocation that maximizes the affine welfare, is NP-Hard~\citep{rothkopf1998computationally}. 
Secondly, the revenue function of AMAs is \emph{non-differentiable} with respect to the AMA parameters.
This issue, which we will discuss in detail in \cref{sec:opt}, arises because the computation of revenue involves searching for the winning allocation among all $(n+1)^m$ deterministic allocations.

To address the scalability and non-differentiability issues of AMA-based automated deterministic auction design, we introduce OD-VVCA, an objective decomposition approach for designing Virtual Valuations Combinatorial Auctions (VVCAs)~\citep{VVCA2005, VVCA2015}. 
Firstly, we constrain the mechanism to be deterministic VVCAs, a subset of deterministic AMAs that boost each bidder-bundle pair, reducing the parameter space from $O((n+1)^m)$ to $O(n\times 2^m)$. 
Secondly, since determining the winning allocations of VVCAs remains NP-Hard~\citep{VVCA2005}, we develop a parallelizable dynamic programming algorithm to solve the problem relatively efficiently. 
Thirdly, we decompose the objective function, i.e., the revenue of the VVCA, into continuous and discontinuous components.
The continuous component is differentiable almost everywhere, allowing us to take its derivative directly.
The discontinuous component is a piecewise constant function, and we use Monte Carlo estimation to approximate the gradient of its Gaussian smoothing.
Finally, by integrating the gradients from both components, we optimize the VVCA parameters using gradient ascent.

Extensive experiments show that OD-VVCA can achieve higher revenue than deterministic baselines and even randomized AMA-based approaches, particularly in large-scale asymmetric auctions. 
Additionally, we conduct an ablation study and case studies to showcase the significant benefits that the objective decomposition brings to \name\ in terms of revenue and optimization. 
All these experimental results indicate the effectiveness of our approach.

\section{Related Work}

Research in automated auction design can be broadly categorized into two threads. 
The first thread, initiated by RegretNet~\citep{dutting2019optimal}, utilizes ex-post regret as a metric to quantify the extent of DSIC violation. 
Building upon this concept, subsequent works such as~\citet{curry2020certifying, peri2021preferencenet, rahme2021auction, rahme2021permutation, curry2022learning, duan2022context, ivanov2022optimal} represent auction mechanisms through neural networks. 
These models aim to achieve near-optimal and approximate DSIC solutions through adversarial training.
The advantages of these regret-based methods lie in their ability to attain high revenue, considering a broad class of auction mechanisms. 
Furthermore, this methodology is generalizable and applicable to various other mechanism design problems, including multi-facility location~\citep{golowich2018deep}, two-sided matching~\citep{ravindranath2021deep}, and data markets~\citep{ravindranath2023data}. 
However, the regret-based methods are not guaranteed to be DSIC, and computing the regret term can be time-consuming.
Moreover, all the regret-based methods are designed for randomized auctions.

The second thread is based on Affine Maximizer Auctions (AMAs)~\citep{AMA}, a weighted variation of the Vickrey-Clarke-Groves (VCG) mechanism.
AMAs assign weights to each bidder and boost each candidate allocation, enabling them to achieve higher revenue than VCG while maintaining DSIC. 
This method is particularly well-suited for scenarios with limited candidate allocations~\citep{li2024deep}.
Furthermore, the sample complexity of AMA has been characterized by \citet{balcan2016sample, balcan2018general, balcan2023generalization}. 
However, conventional AMAs encounter scalability challenges in multi-item auctions with $n$ bidders and $m$ items, given the total of $(n+1)^m$ deterministic candidate allocations. 
Additionally, an AMA's revenue is non-differentiable regarding its parameters, as the computation involves searching for the affine welfare-maximizing allocation.
One approach to address these challenges is to restrict the size of candidate allocations to $s \ll (n+1)^m$ and make the candidate allocations learnable~\citep{curry2022differentiable, duan2023scalable}, leading to randomized candidate allocations.
However, the effective size of candidate allocations still grows with an increasing number of bidders and items. 
As we will discuss in \cref{sec:exp}, this growth results in suboptimal performance in large-scale auctions.

An additional strategy to address the scalability issue of AMA involves further constraining the mechanism to specific subsets, such as Virtual Valuations Combinatorial Auctions (VVCAs)~\citep{likhodedov2004methods, VVCA2005, VVCA2015}, $\lambda$-auctions~\citep{jehiel2007mixed}, mixed bundling auctions~\citep{tang2012mixed}, and bundling boosted auctions~\citep{balcan2021learning}. 
Among these options, VVCAs constitute the most significant subset.
Given this, we focus on restricting the mechanism to be VVCAs. 
However, VVCAs still suffer from the non-differentiability of revenue. 
Existing VVCA-based algorithms~\citep{VVCA2005,VVCA2015} compute the ``gradient'' of revenue based on fixed winning allocations, neglecting the influence of VVCA parameters on the winning allocations.

\section{Preliminary}

\paragraph{Sealed-Bid Auction}
A sealed-bid auction involves $n$ bidders denoted as $[n] = \{1, 2, \dots, n\}$ and $m$ items denoted as $[m] = \{1, 2, \dots, m\}$. 
Each bidder $i$ assigns a valuation $v_i(\Ss)$ to every bundle of items $\Ss \subseteq [m]$ and submits bids for each bundle as $\bmb_i \in \RR^{2^m}$. 
We assume $v_i(\emptyset) = 0$ for all bidders. 
In the additive setting, $v_i(\Ss) = \sum_{j\in\Ss} v_i(\{j\}) = \sum_{j\in\Ss} v_{ij}$.
The valuation profile $V = (\bmv_1, \bmv_2, \dots, \bmv_n)$ is generated from a distribution $D$. The auctioneer lacks knowledge of the true valuation profile $V$ but can observe the bidding profile $B = (\bmb_1, \bmb_2, \dots, \bmb_n)$.

\paragraph{Auction Mechanism}
An auction mechanism $(g, p)$ consists of an allocation rule $g$ and a payment rule $p$. 
Given the bids $B$, $g_i(B)$ computes the allocation result for bidder $i$, which can be a bundle of items or a probability distribution over all the bundles. 
The payment $p_i(B) \geq 0$ computes the price that bidder $i$ needs to pay. 
Each bidder aims to maximize her utility, defined as $u_i(\bmv_i, B) \coloneqq v_i(g_i(B)) - p_i(B)$.
Bidders may misreport their valuations to gain an advantage. 
Such strategic behavior among bidders could make the auction result hard to predict. 
Therefore, we require the auction mechanism to be \emph{dominant strategy incentive compatible (DSIC)}, meaning that for each bidder $i \in [n]$, reporting her true valuation is her optimal strategy regardless of how others report. 
Formally, a DSIC mechanism satisfies that, for any bidder $i \in [n]$ and any bids of others $B_{-i} = (\bmb_1, \dots, \bmb_{i-1}, \bmb_{i+1}, \dots, \bmb_n)$, we have $u_i(\bmv_i,(\bmv_i,B_{-i})) \ge u_i(\bmv_i,(\bmb_i,B_{-i}))$ for arbitrary misreport $\bmb_i$.
Furthermore, the auction mechanism needs to be \emph{individually rational (IR)}, ensuring that truthful bidding results in a non-negative utility for each bidder, that is, $u_i(\bmv_i,(\bmv_i,B_{-i})) \ge 0$. 

\paragraph{Affine Maximizer Auction (AMA)}
AMAs~\citep{AMA} is a generalized version of VCG auctions~\citep{VCG} and inherently ensures DSIC and IR. 
A (deterministic) AMA contains positive weights $w_i \in \RR_+$ for each bidder and boosts $\lambda(A) \in \RR$ for each candidate allocation $A \in \Aa$, where $\Aa$ is the set of all the $(n+1)^m$ deterministic allocations. 
Given bids $B$, an AMA of parameters $(\bw, \bl)$ selects the allocation $A^*$ that maximizes the affine welfare (with an arbitrary tie-breaking rule):
\begin{equation}
\label{eq:allocation}
\begin{aligned}
    A^* = g(B; \bmw, \bm\lambda) \coloneqq \argmax_{A \in \Aa} \sum_{i=1}^n w_iv_i(A) + \lambda(A),
\end{aligned}
\end{equation}
and each bidder $i$ pays for her normalized negative affine welfare impact on other bidders:
\begin{equation}
\label{eq:payment}
\begin{aligned}
    p_i(B; \bmw, \bm\lambda) = \frac{1}{w_i}\max_{A \in \Aa}\left(\sum_{j \neq i} w_j b_j(A) + \lambda(A)\right) - \frac{1}{w_i}\left(\sum_{j \neq i} w_j b_j(A^*) + \lambda(A^*)\right).
\end{aligned}
\end{equation}

\paragraph{Virtual Valuations Combinatorial Auction (VVCA)}
VVCAs~\citep{VVCA2005} is a subset of AMAs, distinguishing itself by decomposing the boost variable of AMAs into $n$ parts, one for each bidder. 
Formally, a VVCA boosts per bidder-bundle pair as 
\begin{equation}
\label{eq:VVCA}
\tag{VVCA}
    \lambda(A) = \sum_{i=1}^n \lambda_i(A_i).
\end{equation}
This decomposition results in a reduction of VVCA parameters from $O((n+1)^m)$ to $O(n\times 2^m)$, aligning with the same order as the input valuation $V$.

\section{Methodology}

In this section, we outline the methodology of \name, as depicted in \cref{fig:OD-VVCA}. 
We begin by formalizing our problem of automated mechanism design for deterministic Virtual Valuations Combinatorial Auctions (VVCAs). 
Next, we introduce the dynamic programming algorithm for computing the auction results of VVCAs and explain how we optimize the VVCA parameters using objective decomposition. 

\subsection{Problem Formulation}

We aim to discover a high-revenue deterministic auction mechanism that satisfies both DSIC and IR for a sealed-bid auction with $n$ bidders and $m$ items. 
Since there is currently no known characterization of a DSIC multi-item combinatorial auction~\citep{dutting2023optimal}, a common strategy is to limit the auction class to AMAs~\citep{VVCA2005,VVCA2015,curry2022differentiable,duan2023scalable}. 
AMAs are inherently DSIC and IR, encompassing a broad class of mechanisms~\citep{lavi2003towards}. 
However, the boost parameters $\bm\lambda$ are defined across the entire candidate allocations, resulting in $O((n+1)^m)$ parameters. 
To address the scalability challenges stemming from the parameter space of AMAs, we further constrain the auction class to be (deterministic) Virtual Valuations Combinatorial Auctions (VVCAs). 
Based on this constraint, we can formalize the auction design problem as the following optimization problem:
\begin{equation}
\label{eq:opt}
\begin{aligned}
    \max_{\bmw \in \RR_+^n, \bm\lambda \in \RR^{n\times 2^m}}~\rev_D(\bmw, \bm\lambda) \coloneqq& \EE_{V \sim D} \left[ \rev(V, \bmw, \bm\lambda) \coloneqq \sum_{i=1}^n p_i(V; \bmw, \bm\lambda) \right],
\end{aligned}
\end{equation}
where we define $R(V, \bmw, \bm\lambda)$ as the revenue of valuation profile $V$ under VVCA parameters $\bmw$ and $\bm\lambda$, which can be further derived as:
\begin{equation}\label{eq:revenue}
    \begin{aligned}
        &\rev(V, \bmw, \bm\lambda) = \overbrace{\sum_{i=1}^n v_i(g(V;\bmw, \bm\lambda))}^{Z(V, \bmw, \bm\lambda)} + \\
        &~~\underbrace{\sum_{i=1}^n \frac{1}{w_i} \max_{A \in \Aa}\left(\sum_{j \neq i} w_j v_j(A) + \sum_{k=1}^n\lambda_k(A_k)\right) - \sum_{i=1}^n\frac{1}{w_i} \max_{A \in \Aa} \left(\sum_{j = 1}^n w_j v_j(A) + \lambda_j(A_j)\right)}_{F(V, \bmw, \bm\lambda)},
    \end{aligned}
\end{equation}
where we denote $Z(V, \bmw, \bm\lambda)$ and $F(V, \bmw, \bm\lambda)$ as the two components of $R(V, \bmw, \bm\lambda)$.

\subsection{Winning Allocation Determination}\label{sec:method:DP}

As shown in \cref{eq:revenue}, computing $\rev(V, \bmw, \bm\lambda)$ requires finding the winning allocation $g(V; \bmw, \lambda)$ and the allocation that maximizes the affine welfare for each bidder $i \in [n]$, which can be viewed as the allocation maximizing the valuation profile $(\bm{0}, V_{-i})$. 
Therefore, we need to develop an algorithm to compute the winning allocation for an arbitrary given valuation profile.

The winning allocation determination of a VVCA is NP-Hard~\citep{VVCA2005}. 
To address this challenge as efficiently as possible, we propose a parallelizable dynamic programming algorithm for computing the winning allocation of VVCAs.
Firstly, we denote $\MAW(i, \Ss)$ as the Maximum Affine Welfare if we allocate all items in bundle $\Ss \in 2^m$ to the first $i$ bidders, and we use $\AB(i, \Ss)$ to record the Allocated Bundle of bidder $i$. 
We initialize $\MAW(1, \Ss) = w_1v_1(\Ss) + \lambda_1(\Ss)$ and $\AB(1, \Ss) = \Ss$, and compute $\MAW(i, \Ss)$ and $\AB(i, \Ss)$ for $i = 2, 3, \dots, n$ and $\Ss \in 2^m$ iteratively. 
The computation of $\MAW(i, \Ss)$ and $\AB(i, \Ss)$ involves enumerating all subsets of $\Ss$ as the allocated bundle to bidder $i$, which are given by:
\begin{equation}
\label{eq:DP}
\begin{aligned}
    \MAW(i, \Ss) &= \max_{\Bb \subseteq \Ss} \MAW(i - 1, \Ss\backslash \Bb) + w_iv_i(\Bb) + \lambda_i(\Bb), \\
    \AB(i, \Ss) &= \argmax_{\Bb \subseteq \Ss} \MAW(i - 1, \Ss\backslash \Bb) + w_iv_i(\Bb) + \lambda_i(\Bb),
\end{aligned}
\end{equation}
which can be computed in parallel for all subsets $\Bb$ using vector operations (see \cref{app:DP} for details). 
Finally, to obtain the winning allocation, we define $\Ss_n \coloneqq \argmax_{\Ss \subseteq [m]} \MAW(n, \Ss)$ as the set of all allocated items in the affine welfare-maximizing allocation, and iteratively determine the allocated bundle $\AB(i, \Ss_i)$ for bidder $i = n, n - 1, \dots, 1$, updating $\Ss_{i - 1} \coloneqq \Ss_i - \AB(i, \Ss_i)$.

The pseudocode of the entire DP algorithm is provided in \cref{app:DP}, where we also prove that the time complexity and space complexity of the DP algorithm are $O(n3^m)$ and $O(n2^m)$, respectively. 
Notably, the input valuation profile is already $O(n2^m)$, demonstrating the efficiency of our dynamic programming algorithm. 
Furthermore, as we discussed, the computation of \cref{eq:DP} can be performed using vector operations.
Therefore, when solving the DP for a batch of valuations, we can further accelerate the computation through matrix operations, which can be parallelized using GPUs.

\begin{figure}[t]
    \centering
    \includegraphics[width = \linewidth, trim=0 0 0 0,clip]{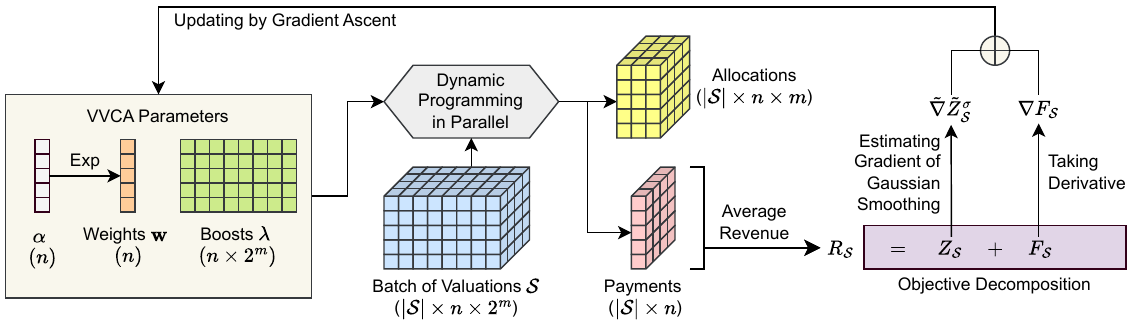}
    \caption{
        Overview of OD-VVCA for automated design of deterministic auctions. 
        We constrain the mechanism to VVCAs with $n + n \times 2^m$ parameters, introducing $\bm\alpha \in \mathbb{R}^n$ such that $\bmw = e^\ba \in \mathbb{R}_+^n$ to accommodate the positive range of $\bmw$.
        A parallelizable dynamic programming algorithm is employed to compute the allocations and payments for a batch of valuations.
        Afterward, the average revenue objective is decomposed into two components, and their gradients are estimated or computed directly. 
        By integrating these gradients, we update the VVCA parameters through gradient ascent.
    }    
    \label{fig:OD-VVCA}
\end{figure}

Based on the DP algorithm, we can efficiently compute $\rev(V, \bmw, \bm\lambda)$ as described in \cref{eq:revenue}. 
The computation takes $n+1$ DPs: one to solve the winning allocation for valuation profile $V$, and $n$ additional DPs to solve for the valuation profile $(\bm{0}, V_{-i})$ for each bidder $i \in [n]$. 

\subsection{Optimization}\label{sec:opt}

To optimize the VVCA parameters $\bmw$ and $\bm\lambda$ in \cref{eq:opt}, we first introduce $\bm\alpha \in \RR^n$ as the logarithm of $\bmw$ (such that $\bmw = e^\ba \in \RR_+^n$) and optimize $\bm\alpha$ instead to handle the positive range of $\bmw$. 
Then we optimize the empirical revenue, which is defined as:
\begin{equation}\label{eq:revenue:empirical}
\begin{aligned}
    \rev_\Ss(\bmw = e^\ba, \bm\lambda) \coloneqq \frac{1}{|\Ss|}\sum_{V \in \Ss} \rev(V,e^\ba, \bl) = \underbrace{\frac{1}{|\Ss|}\sum_{V \in \Ss}Z(V, e^\ba, \bm\lambda)}_{Z_\Ss(e^\ba, \bl)} + \underbrace{\frac{1}{|\Ss|}\sum_{V \in \Ss}F(V, e^\ba, \bm\lambda)}_{F_\Ss(e^\ba, \bl)},
\end{aligned}
\end{equation}
where $\Ss$ is a dataset of valuations sampled i.i.d. from $D$.
However, $\rev_\Ss(e^\ba, \bl)$ is non-differentiable with respect to $\ba$ and $\bl$, since the computation of $\rev(V,e^\ba, \bl)$ involves finding the winning allocation from all deterministic candidate allocations and computing the social welfare of the winning allocation (i.e., $Z(V, e^\ba, \bl)$ in \cref{eq:revenue}).
Therefore, we decompose it into two parts: $F_\Ss(e^\ba, \bl)$ and $Z_\Ss(e^\ba, \bl)$, and optimize them separately.

The optimization of $F_\Ss(e^\ba, \bm\lambda)$ is straightforward since it is a continuous and piece-wise non-constant linear function with respect to $e^\ba$ and $\bm\lambda$. 
Thus, we can easily optimize $F_\Ss(e^\ba, \bm\lambda)$ through its gradient $\nabla F_\Ss(e^\ba, \bm\lambda)$, where the gradient of the maximum function can be estimated similarly to the ReLU function.

However, optimizing $Z_\Ss(e^\ba, \bm\lambda)$ is more challenging. 
This is because $Z_\Ss(e^\ba, \bm\lambda)$ is a discontinuous, piece-wise constant function. 
Its derivatives at differentiable points are always $0$, and it undergoes sudden changes at the non-differentiable discontinuous points, making it not feasible to optimize $Z_\Ss(e^\ba, \bm\lambda)$ using its own gradient.
To address this issue, similar to \citet{bichler2021learning}, we optimize $Z_\Ss(e^\ba, \bm\lambda)$ using Gaussian smoothing techniques.
Specifically, we optimize the Gaussian smoothing approximation of $Z_\Ss(e^\ba, \bm\lambda)$, which is defined as:
\begin{equation*}
    \Tilde{Z}_\Ss^\sigma(e^{\bm\alpha}, \bm\lambda) \coloneqq \EE_{\bm\epsilon \sim N(0, 1)^n, \bm\delta \sim N(0, 1)^{n\times 2^m}} \left[Z_\Ss(e^{\bm\alpha + \sigma\bm\epsilon}, \bm\lambda + \sigma\bm\delta)  \right],
\end{equation*}
where $\sigma > 0$ and $N(0, 1)$ is the standard normal distribution. 
Clearly, as $\sigma \to 0$, $\Tilde{Z}_\Ss^\sigma(e^{\bm\alpha}, \bm\lambda)$ approaches $Z_\Ss(e^\ba, \bm\lambda)$.
As for the differentiability property of $\Tilde{Z}_\Ss^\sigma(e^{\bm\alpha}, \bm\lambda)$, we have the following proposition:
\begin{restatable}{proposition}{propositionZSmooth}
    \label{proposition:Z:smooth}
    $\Tilde{Z}_\Ss^\sigma(e^{\bm\alpha}, \bm\lambda)$ is continuous and differentiable with derivatives:
    \begin{equation*}
    \begin{aligned}
        \nabla_{\ba} \Tilde{Z}_\Ss^\sigma(e^{\bm\alpha}, \bm\lambda) &= \EE_{\bm\epsilon \sim N(0, 1)^n, \bm\delta \sim N(0, 1)^{n\times 2^m}} \left[\frac{Z_\Ss(e^{\bm\alpha + \sigma\bm\epsilon}, \bm\lambda + \sigma\bm\delta) - Z_\Ss(e^{\bm\alpha}, \bm\lambda)}{\sigma}\bm\epsilon  \right]\\
        \nabla_{\bl} \Tilde{Z}_\Ss^\sigma(e^{\bm\alpha}, \bm\lambda) &= \EE_{\bm\epsilon \sim N(0, 1)^n, \bm\delta \sim N(0, 1)^{n\times 2^m}} \left[\frac{Z_\Ss(e^{\bm\alpha + \sigma\bm\epsilon}, \bm\lambda + \sigma\bm\delta) - Z_\Ss(e^{\bm\alpha}, \bm\lambda)}{\sigma}\bm\delta  \right].
    \end{aligned}
    \end{equation*}
\end{restatable}
According to \cref{proposition:Z:smooth}, we can unbiasedly estimate the gradient of $\Tilde{Z}_\Ss^\sigma(e^{\bm\alpha}, \bm\lambda)$ by Monte Carlo sampling, which is done by computing the average of numerical differentiations in $n_r$ random directions generated by Gaussian sampling:
\begin{equation*}
\begin{aligned}
    &\nabla_{\bm\alpha} \tilde{Z}_{\Ss}(e^{\bm\alpha}, \bm\lambda) \approx \tilde{\nabla}_{\bm\alpha} \tilde{Z}_{\Ss}(e^{\bm\alpha}, \bm\lambda) = \frac{1}{n_r}\sum_{i=1}^{n_r} \frac{1}{\sigma}\left(Z_{\Ss}(e^{\bm\alpha + \sigma\bm\epsilon_i}, \bm\lambda + \sigma\bm\delta_i) - Z_{\Ss}(e^{\bm\alpha}, \bm\lambda)\right) \bm\epsilon_i, \\
    &\nabla_{\bm\lambda} \tilde{Z}_{\Ss}(e^{\bm\alpha}, \bm\lambda) \approx \tilde{\nabla}_{\bm\lambda} \tilde{Z}_{\Ss}(e^{\bm\alpha}, \bm\lambda) = \frac{1}{n_r}\sum_{i=1}^{n_r} \frac{1}{\sigma}\left(Z_{\Ss}(e^{\bm\alpha + \sigma\bm\epsilon_i}, \bm\lambda + \sigma\bm\delta_i) - Z_{\Ss}(e^{\bm\alpha}, \bm\lambda) \right) \bm\delta_i,
\end{aligned}
\end{equation*}
where $\bm\epsilon_i \sim N(0, 1)^n$, and $\bm\delta_i \sim N(0, 1)^{n\times 2^m}$. 
It takes totally $n_r$ additional DPs to get the estimated gradient $\tilde{\nabla}{Z}_{\Ss}(e^{\bm\alpha}, \bm\lambda)$, based on which we can optimize $\Tilde{Z}_{\Ss}^\sigma(e^{\bm\alpha}, \bm\lambda)$ through gradient ascent.
This approach offers an approximate method for optimizing $Z_\Ss(e^{\bm\alpha}, \bm\lambda)$.

\begin{remark}
    It is also feasible to optimize $R_\Ss(e^\ba, \bl)$ using the same Gaussian smoothing technique as $Z_\Ss(e^\ba, \bl)$. 
    However, this approach would require $(n + 1)n_r$ additional DPs to compute the estimated gradient.
    In comparison, our objective decomposition reduces the number of additional DPs to $n_r$, enhancing efficiency significantly.
\end{remark}

\section{Experiments}\label{sec:exp}

In this section, we present empirical experiments to evaluate the effectiveness of \name. 
These experiments are conducted on a Linux machine equipped with $4$ NVIDIA Graphics Processing Units (GPUs), each with a memory size of $11$GB. 
Each result is averaged across $5$ distinct runs. 
The standard deviation of \name\ across all these runs is below 1\%.

\subsection{Setup}\label{sec:exp:setup}

\paragraph{Auction Settings}
We consider a variety of valuation distributions, including both symmetric and asymmetric, as well as additive and combinatorial types. They are listed as follows:
\begin{enumerate}[label=(\Alph*),leftmargin=7mm,itemsep=0mm,topsep=0mm]
    \item \label{settingA} (Symmetric Uniform) For all bidder $i$ and item $j$, the valuation $v_{ij}$ is sampled from $U[0, 1]$.
    The valuation is additive, that is, for all item bundle $S$ and bidder $i$, $v_{i}(\Ss) = \sum_{j \in \Ss} v_{ij}$.
    This setting is widely used in previous researches~\citep{VVCA2015,dutting2019optimal,curry2022differentiable}.
    \item \label{settingB} 
    (Asymmetric Uniform) For all bidder $i$ and item $j$, the valuation $v_{ij}$ is sampled from $U[0, i]$.
    The valuation is additive but asymmetric for bidders.
    \item \label{settingC} 
    (Lognormal) For all bidder $i$ and item $j$, the valuation $v_{ij}$ is sampled from $\mathrm{Lognormal}(0, 1 / i^2)$.
    The valuation is additive but asymmetric for bidders.
    \item \label{settingD} 
    (Combinatorial) For all bidder $i$ and item $j$, the valuation $v_{ij}$ is sampled from $U[1, 2]$.
    For item bundle $\Ss \subseteq [m]$ and bidder $i$, the valuation $v_{i}(\Ss) = \sum_{j \in \Ss} v_{ij} + \epsilon_{i\Ss}$, where $\epsilon_{i\Ss}$ is sampled from $U[- |\Ss| / 2, |\Ss| / 2]$. 
    This setting is also used in \cite{VVCA2015}.
\end{enumerate}

\paragraph{Baselines} 
We compare OD-VVCA against both randomized and deterministic methods.
The randomized methods include:
\begin{enumerate}[leftmargin=7mm,itemsep=0mm,topsep=0mm]
    \item \emph{Lottery AMA}~\citep{curry2022differentiable}, a randomized AMA-based approach that directly sets the candidate allocations, bidder weights, and boost variables as all the learnable weights. 
    \item \emph{AMenuNet}~\citep{duan2023scalable}, similar to Lottery AMA, but uses a transformer-based architecture to optimize the parameters of AMA, hence more powerful and applicable to general contextual auctions. 
\end{enumerate} 
Note that we ignore the regret-based methods~\citep{dutting2019optimal} since they are not strictly IC, and thus not directly comparable to \name. 
The deterministic methods include:
\begin{enumerate}[leftmargin=7mm,itemsep=0mm,topsep=0mm]
    \item VCG~\citep{VCG}, which is the most classical special case of VVCA.
    \item Item-Myerson~\citep{myerson1981optimal}, a strong baseline that independently applies Myerson auction with respect to each item.
    \item BBBVVCA~\citep{VVCA2005,VVCA2015}, the pioneering work of optimizing VVCA parameters, which minimizes the revenue loss for each bidder-bundle by gradient descent. We use the same dynamic programming algorithm as \name\ to compute the winning allocation.
    \item FO-VVCA, an ablated version of \name\ which only optimizes $F_\Ss(e^\ba, \bl)$. 
\end{enumerate} 

\paragraph{Hyperparameters and Implementation}
For \name, we select the number of random directions $n_r$ from $\{8, 16, 32\}$ and the standard deviation $\sigma$ from $\{0.01, 0.003, 0.001\}$. 
The training starts at classic VCG, i.e., $\bmw = 1$ and $\bm\lambda = 0$. 
For Lottery AMA and AMenuNet, since the candidate allocation set cannot be arbitrarily large due to the constraint of computational resources, we select the candidate size from $\{32, 128, 1024, 4096\}$ to balance their performance and computational feasibility. 
Further implementation details of \name\ and baseline methods can be found in \cref{app:implementation}.

\begin{table*}[t]
    \centering
    \caption{
        The revenue results in both small and large settings. 
        For each setting, the highest revenue among all methods is \textbf{bolded}, and the highest revenue among all deterministic methods (the methods except for Lottery AMA and AMenuNet) is \underline{underlined}. 
        It is important to note that we do not report the results of Item-Myerson, Lottery AMA, and AMenuNet for setting~\ref{settingD}. 
        This is because Item-Myerson is not DSIC in a non-additive scenario, and Lottery AMA and AMenuNet are not applicable for combinatorial valuations.
    }
    \vspace{5pt}
    \subtable[Revenue results in small-scale settings.]{
        \resizebox{\textwidth}{!}{
        \begin{tabular}{lcccccccccc}
            \toprule
            \multirow{2}{*}{Method} & \multicolumn{5}{c}{Symmetric} & \multicolumn{5}{c}{Asymmetric} \\
            \cmidrule(l{10pt}r{10pt}){2-6} \cmidrule(l{10pt}r{10pt}){7-11}
            & 2$\times$2\ref{settingA} & 2$\times$5\ref{settingA} & 3$\times$10\ref{settingA} & 2$\times$2\ref{settingD} & 3$\times$10\ref{settingD} &  5$\times$3\ref{settingB} & 3$\times$10\ref{settingB} & 2$\times$5\ref{settingC} & 5$\times$3\ref{settingC} & 3$\times$10\ref{settingC}\\
            \midrule
            Lottery AMA  & \textbf{0.8680} & 2.2354 & 5.3450 & - & - & 6.5904 & 11.9814 & \textbf{5.7001}& 4.0195 & 12.9863 \\
            AMenuNet  & 0.8618 & \textbf{2.2768} & 5.5986 & - & - &  6.7743& 12.3419 &5.6512 & 4.1919 & 13.6094  \\
            \midrule
            VCG & 0.6678 & 1.6691 & 5.0032 & 2.4576 & 15.6305 &  6.0470 & 8.9098 &3.8711& 3.7249 & 10.6495\\
            Item-Myerson & \underline{0.8330} & 2.0755 & 5.3141 & - & - & 5.3909 & 8.9110 &4.4380  & 3.6918 & 9.3740\\
            BBBVVCA & 0.7781 & 2.2576 & 5.7647 & 2.6111 & 16.3504 &  6.6118 & 10.5134 &4.8174 & 4.1879 & 11.4605\\
            FO-VVCA  & 0.7836 & \underline{2.2638} & 5.7876 & 2.6205 & \textbf{\underline{16.3786}} &  6.8783 & 11.8093 & 5.5536& 4.1906 & 13.0144\\
            \name  & 0.8284 & 2.2632 & \textbf{\underline{5.8230}} & \textbf{\underline{2.6802}} & 16.3675 &\textbf{\underline{7.0344}} & \textbf{\underline{12.5497}} & \underline{5.6682}  & \textbf{\underline{4.3289}} & \textbf{\underline{13.6223}} \\
            \bottomrule
        \end{tabular}
        \label{tab:revenue.small}
        }
    }
    \subtable[Revenue results in large-scale settings.]{
        \resizebox{\textwidth}{!}{  
        \begin{tabular}{lcccccccccc}
            \toprule
            \multirow{2}{*}{Method} & \multicolumn{4}{c}{Symmetric} & \multicolumn{6}{c}{Asymmetric} \\
            \cmidrule(l{10pt}r{10pt}){2-5} \cmidrule(l{10pt}r{10pt}){6-11}
            & 5$\times$10\ref{settingA} & 10$\times$5\ref{settingA} & 5$\times$10\ref{settingD} & 10$\times$5\ref{settingD} & 5$\times$10\ref{settingB} & 10$\times$5\ref{settingB} & 30$\times$5\ref{settingB} & 5$\times$10\ref{settingC} & 10$\times$5\ref{settingC} & 30$\times$5\ref{settingC} \\
            \midrule
            Lottery AMA  & 5.5435 & 3.0210 & - & - & 21.4092 & 24.3684 & 74.1259 & 13.0961 & 6.1878 & 4.8812 \\
            AMenuNet  & 6.5210 & 3.8156 & - & - & 21.7912 & 27.4230  & 88.2373 &13.6266& 7.0347& 6.1055  \\
            \midrule
            VCG & 6.6690 & 4.0887 & 17.2580 & \textbf{\underline{9.4405}} &20.1567 & 26.7979& 105.4290 & 12.4163& 6.5472& 6.6054 \\
            Item-Myerson  & 6.7132 & 4.0868 & - & - & 17.9697 & 24.7287 & 102.1878 & 12.3060 & 7.4868 & 7.6173   \\
            BBBVVCA  & 6.9765 & 4.1038 & 17.4657 & 9.4356 & 22.4459 &28.2620 &107.5630& 14.4585 &7.4899&7.4989\\
            FO-VVCA   & \textbf{\underline{6.9904}} & \textbf{\underline{4.1107}} & \textbf{\underline{17.5026}} & 9.4397 & 23.1474 & 28.4363 & 107.8873 &14.6641 & 7.2678 & 7.0965 \\
            \name  & 6.9848 & 4.0830 & 17.4614 & 9.4388 &\textbf{\underline{24.4301}} &     \textbf{\underline{28.9711}} &\textbf{\underline{108.2780}}&\textbf{\underline{14.8071}}& \textbf{\underline{7.5882}} & \textbf{\underline{7.6446}}\\
            \bottomrule
        \end{tabular}
        }
    }
    \label{tab:revenue.large}
    \label{tab:revenue}
\end{table*}

\subsection{Revenue Experiments}

The results of revenue experiments are presented in \cref{tab:revenue}, where we use the notation $n \times m (X)$ to denote an auction with $n$ bidders and $m$ items of setting $(X)$. 
All the machine learning-based methods take comparable training time, which we discuss in \cref{app:time}.

Among the deterministic baselines, \name\ outperforms them in most settings, especially in large-scale asymmetric scenarios. 
Specifically, the comparison between \name\ and VCG demonstrates that integrating affine parameters into VCG significantly enhances revenue, and the comparison between \name\ and Item-Myerson indicates the strong revenue performance of \name. 
OD-VVCA is also distinctly advantageous over BBBVVCA, which optimizes VVCA parameters under a fixed winning allocation and overlooks the impact of these parameters on allocation outcomes.
Additionally, \name\ consistently outperforms FO-VVCA across various asymmetric and most symmetric settings, underscoring the efficacy of the Gaussian smoothing technique to optimize $Z_\Ss(e^\ba, \bl)$.

Among the randomized baselines, Lottery AMA and AMenuNet both demonstrate commendable performance in smaller settings. 
However, in larger settings, as shown in \cref{tab:revenue.large}, \name\ consistently outperforms these two randomized AMA methods. 
This outcome suggests that within the constraint of computational resource limitations, \name\ is more effective than Lottery AMA and AMenuNet, particularly in handling larger settings.

\subsection{Case Study}

\begin{figure}[t]
    \centering
    \includegraphics[width = \linewidth, trim=0 0 0 0,clip]{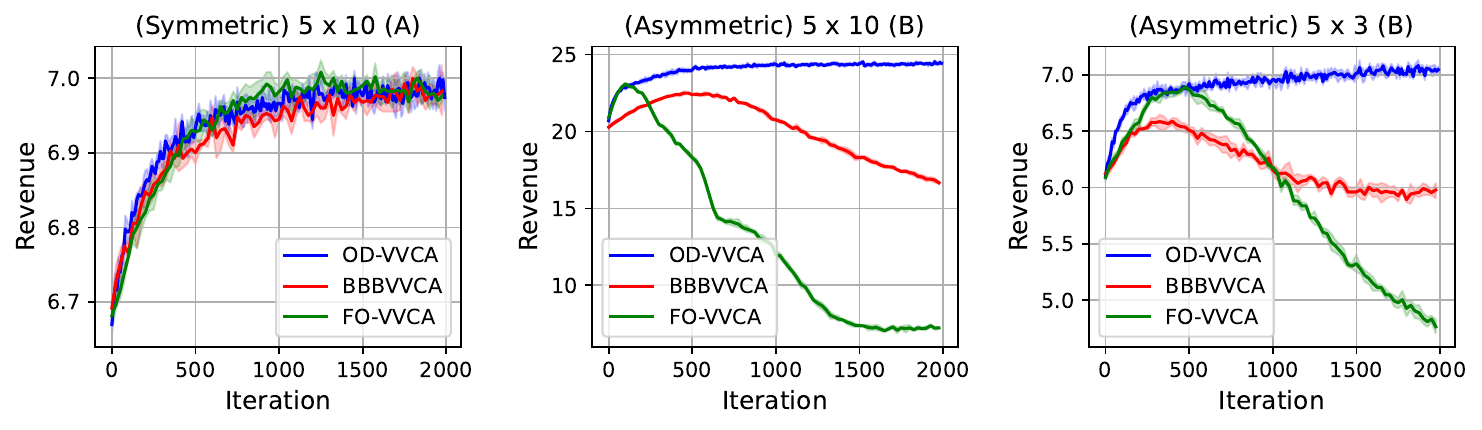}
    \caption{
        Training curves for OD-VVCA, BBBVVCA, and FO-VVCA under $5 \times 10$\ref{settingA}, $5 \times 10$\ref{settingB} and $5\times 3$\ref{settingB}. 
        The plots depict average results across $5$ runs and the $95\%$ 
    }
    \label{fig:loss}
\end{figure}

To further illustrate the advantages of \name\ over BBBVVCA and FO-VVCA, we present the revenue trajectories of them during training in symmetric ($5\times 10$)\ref{settingA} and asymmetric ($5\times 10$\ref{settingB} and $5\times 3$\ref{settingB}) scenarios in \cref{fig:loss}. 
We can see that while all three methods converge to similar revenue in the symmetric $5\times 10$\ref{settingA} scenario, the revenue of BBBVVCA and FO-VVCA declines significantly in the two asymmetric settings. 
The instability of BBBVVCA highlights the drawbacks of neglecting the influence of VVCA parameters on winning allocation. 
Meanwhile, the comparison between OD-VVCA and OD-VVCA indicates the effectiveness and stability of incorporating the Gaussian smoothing technique in optimizing VVCA parameters.

\begin{figure*}[t]
    \centering
    \includegraphics[width=\linewidth]{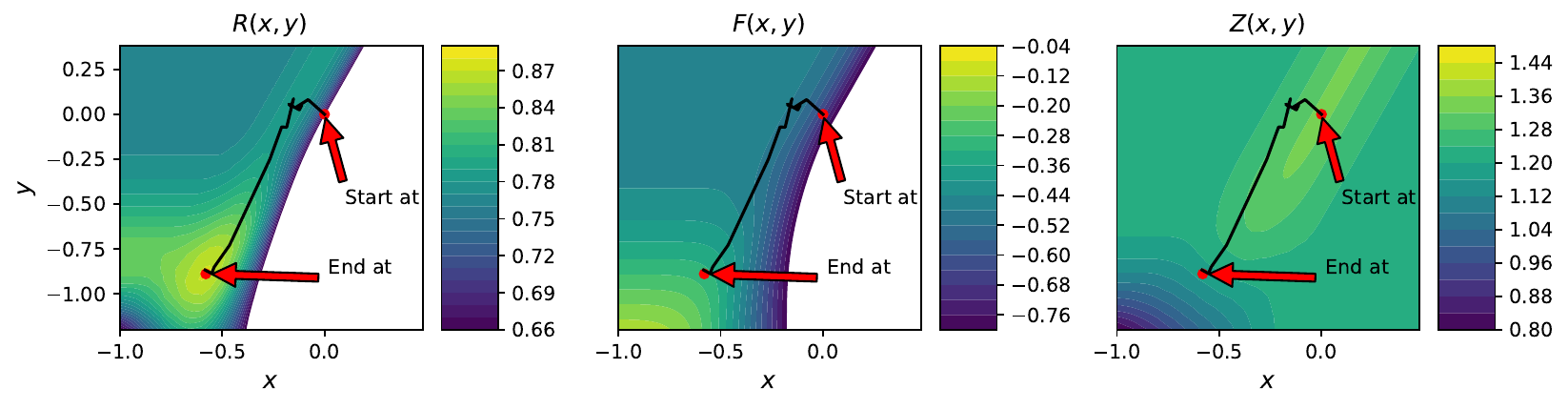}
    \caption{
        Training dynamics of OD-VVCA over $2000$ iterations in $2\times 2$ (\ref{settingA}). 
        The trajectory of $(x, y)$ is simultaneously plotted on $R(x, y)$, $F(x, y)$, and $Z(x, y)$. 
        The trajectory converges at a point balancing the values of both $F(x, y)$ and $Z(x, y)$, illustrating the effectiveness of optimizing $R(x, y)$ through objective decomposition.
    }
    \label{fig:dynamics}
\end{figure*}

To visualize the training process of OD-VVCA, we adopt a symmetric simplification for bidders and items in $2\times 2$ (\ref{settingA}). 
Specifically, we fix the bidder weights $w_1=w_2=1$ and $\lambda_1(\emptyset)=\lambda_2(\emptyset)=0$, then we set other boosts as: 
\begin{equation*}
   \lambda_1(\{1\})=\lambda_2(\{1\})=\lambda_1(\{2\})=\lambda_2(\{2\})=x, \quad \lambda_1(\{1, 2\})=\lambda_2(\{1, 2\})=y.
\end{equation*}
Under this simplification, when fixing a dataset of valuation $\mathcal{S}$, we use $R(x, y)$, $F(x, y)$, and $Z(x, y)$ to denote $R_{\mathcal{S}}(\bm{w}, \bm{\lambda})$, $F_{\mathcal{S}}(\bm{w}, \bm{\lambda})$, $Z_{\mathcal{S}}(\bm{w}, \bm{\lambda})$, respectively. 
In \cref{fig:dynamics}, we illustrate the trajectory of $(x, y)$ concurrently on $R(x, y)$, $F(x, y)$, and $Z(x, y)$. 
Here, $R(x, y)= F(x, y)+ Z(x, y)$ according to \cref{eq:revenue:empirical}. 
The trajectory initiates from a VCG auction (i.e., $x = y = 0$), and converges to an improved solution that balances the values of both $F(x, y)$ and $Z(x, y)$. 
This outcome underscores the efficacy of our Gaussian smoothing technique to optimize $Z(x, y)$.

\begin{figure}[t]
    \centering
    \includegraphics[width=\linewidth]{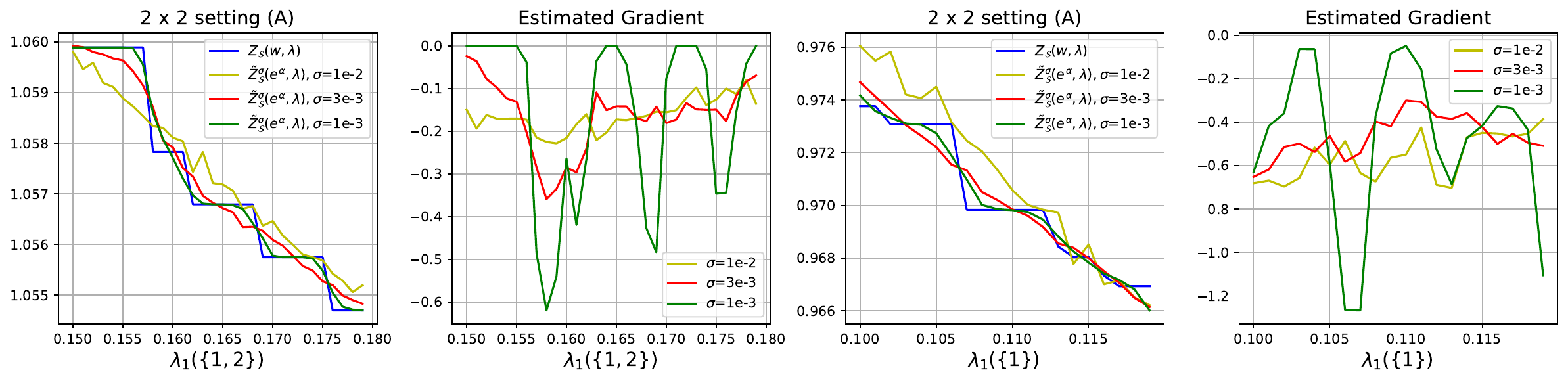}
    \caption{
        Visualization of $Z_{\Ss}(e^\ba, \bl)$, $\tilde{Z}^\sigma_{\Ss}(e^\ba, \bl)$, and the corresponding estimated gradient $\tilde\nabla \tilde{Z}_{\Ss}(e^\ba, \bl)$ in $2\times 2$\ref{settingA} with respect to one parameter. 
    }
    \label{fig:gradient}
\end{figure}

To further illustrate the impact of the Gaussian smoothing technique in OD-VVCA, we explore a $2\times2$ scenario (\ref{settingA}), adjusting only one of the boost variables $\lambda$ while maintaining the remaining VVCA parameters fixed. 
We plot $Z_{\mathcal{S}}(\bm{w}, \bm{\lambda})$ alongside the computed $\tilde{Z}^{\sigma}_{\mathcal{S}}(\bm{w}, \bm{\lambda})$ for three distinct $\sigma$ values in \cref{fig:gradient}. 
The figure shows that $Z_{\mathcal{S}}(\bm{w}, \bm{\lambda})$ experiences non-differentiability at several points. 
The Gaussian smoothing technique addresses this by smoothing $Z_{\mathcal{S}}(\bm{w}, \bm{\lambda})$, with varying smoothness achieved by different $\sigma$ values. 
Specifically, a smaller $\sigma$ provides a closer approximation of $\tilde{Z}^{\sigma}_{\mathcal{S}}(\bm{w}, \bm{\lambda})$ to $Z_{\mathcal{S}}(\bm{w}, \bm{\lambda})$, and a larger $\sigma$ results in a less precise approximation but more effectively captures the general trend with a flatter gradient. 
Though these insights emphasize the criticality of careful parameter selection in OD-VVCA, they do not necessarily imply difficulty in tuning the hyperparameters. 
In practice, our experiments indicate that $\sigma$ values within $\{0.01, 0.003, 0.001\}$ are generally well-suited for most scenarios.

\section{Conclusion and Future Work}\label{sec:conclusion}

In this paper, we introduce OD-VVCA, an objective decomposition approach for designing deterministic Virtual Valuation Combinatorial Auctions (VVCAs). 
Our method employs a parallelizable dynamic programming algorithm to compute the winning allocation and revenue in a VVCA efficiently. 
Subsequently, we decompose the objective function of revenue into continuous and discontinuous components.
The continuous portion can be directly optimized by taking derivatives.   
For the discontinuous part, we use Monte Carlo estimation to approximate the gradient of its Gaussian-smoothed version. 
Finally, we demonstrate the efficacy, stability, and scalability of OD-VVCA in various auction settings through experimental results.

As for future work, while OD-VVCA is tailored for general combinatorial auctions, exploring its adaptation for industrial applications presents an intriguing avenue. 
Additionally, given that the Gaussian smoothing technique requires extra iterations to approximate the gradient via Monte Carlo sampling, it would be valuable to investigate more efficient methods for gradient approximation.
	
	\bibliographystyle{plainnat}
	\bibliography{reference}
	
	\appendix
	\newpage
	\appendix
\section{About the Winning Allocation Deterministic Dynamic Programming} \label{app:DP}

In this section, we further discuss our proposed dynamic programming (DP) algorithm for winning allocation determination in \cref{sec:method:DP}.

\subsection{The Parallel implementation of the DP}

For a single valuation $V$, given the bidder $i$ and the item bundle $\mathcal{S}$, we can compute $\mathrm{MAW}(i, S)$ and $\mathrm{AB}(i, S)$ according to \cref{eq:DP} using vector operations, which can be parallelized by the GPU.
Firstly, for all item bundle $\mathcal{S} \subseteq [m]$, we denote the set of all subsets of $\mathcal{S}$ as $\overrightarrow{\mathcal{B}(\mathcal{S})} \coloneqq (\mathcal{B})_{\mathcal{B} \subseteq \mathcal{S}}$, which is a $2^{|\mathcal{S}|}$-dimensional vector. 
On top of that, we can compute the intermediate result $\mathrm{Int}(i, \mathcal{S}, \overrightarrow{\mathcal{B}(\mathcal{S})}) = (\mathrm{Int}(i, \mathcal{S}, \Bb))_{\Bb \in \overrightarrow{\Bb(\Ss)}}$ by vector operations:
\begin{equation*}
    \mathrm{Int}(i, \mathcal{S}, \overrightarrow{\mathcal{B}(\mathcal{S})}) = \mathrm{MAW}(i - 1, \mathcal{S} - \overrightarrow{\mathcal{B}(\mathcal{S})}) + w_i v_i(\overrightarrow{\mathcal{B}(\mathcal{S})}) + \lambda_i(\overrightarrow{\mathcal{B}(\mathcal{S})})
\end{equation*}
Subsequently, we can obtain $\mathrm{MAW}(i, \mathcal{S})$ and $\mathrm{AB}(i, \mathcal{S})$ by computing the maximum element of the vector $\mathrm{Int}(i, \mathcal{S}, \overrightarrow{\mathcal{B}(\mathcal{S})})$:
\begin{equation*}
\begin{aligned}
    \MAW(i, \Ss) = \max_{\Bb \in \overrightarrow{\Bb(\Ss)}} \mathrm{Int}(i, \Ss, \Bb), \quad \AB(i, \Ss) = \argmax_{\Bb \in \overrightarrow{\Bb(\Ss)}} \mathrm{Int}(i, \Ss, \Bb)
\end{aligned}
\end{equation*}
It is evident that all these operations are vector-based.

In summary, we provide the pseudocode of the parallel DP implementation for a single valuation in \cref{alg:DP:parallel}.
For the valuations in a minibatch $\Ss = {V^1, V^2, \dots, V^\Ss}$, we can compute all the DP computations through matrix operations. 
The reason is that, as discussed above, the DP computation for a single valuation involves vector operations. 
Thus, the computation for all valuations can be executed using matrix operations.

\begin{algorithm}[h]
    \caption{Winning Allocation Determination by Dynamic Programming}
    \label{alg:DP:parallel}
    \begin{algorithmic}[1]
        \STATE {\bfseries Input:} Valuation $V$, VVCA parameters $\bmw \in \RR_+^n$ and $\bm\lambda \in \RR^{n2^m}$. 
        \STATE {\bfseries Output:} Maximum affine welfare $\MAW^*$ and the corresponding allocation $A^*$.  \label{alg:DP:for:dp}
        
        \FOR {$\Ss \subseteq [m]$}\label{alg:DP:for:init}
            \STATE Initialize $\MAW(1, \Ss) \gets w_1v_1(\Ss) + \lambda_1(\Ss)$.
            \STATE Initialize $\AB(1, \Ss) \gets \Ss$.
            \STATE Initialize $\overrightarrow{\Bb(\Ss)} \gets (\Bb)_{\Bb \subseteq \Ss}$.
        \ENDFOR
        
        \FOR {$i = 2$ \textbf{to} $n$}
            \FOR {$\Ss \subseteq [m]$}
                \STATE Compute the intermediate result: \\
                \quad $\mathrm{Int}(i, \mathcal{S}, \overrightarrow{\mathcal{B}(\mathcal{S})}) = \mathrm{MAW}(i - 1, \mathcal{S} - \overrightarrow{\mathcal{B}(\mathcal{S})}) + w_i v_i(\overrightarrow{\mathcal{B}(\mathcal{S})}) + \lambda_i(\overrightarrow{\mathcal{B}(\mathcal{S})})$
                \STATE Compute the maximum affine welfare when allocating $\Ss$ to the previous $i$ bidders: \\
                \quad $\MAW(i, \Ss) = \max_{\Bb \in \overrightarrow{\Bb(\Ss)}} \mathrm{Int}(i, \Ss, \Bb)$
                \STATE Record the allocated bundle of bidder $i$: \\
                \quad $\AB(i, \Ss) \gets \argmax_{\Bb \in \overrightarrow{\Bb(\Ss)}} \mathrm{Int}(i, \Ss, \Bb)$
            \ENDFOR
        \ENDFOR
        \STATE $\Ss_n \gets \argmax_{\Ss \subseteq [m]} \MAW(n, \Ss)$.
        \FOR {$i = n$ \textbf{to} $1$} \label{alg:DP:for:alc}
            \STATE $A^*_i \gets \AB(i, \Ss_i)$, the allocated bundle of bidder $i$.
            \STATE $\Ss_{i-1} \gets \Ss_i - A^*_i$
         \ENDFOR
    
        \STATE \textbf{return} Winning Allocation $A^* = (A^*_1, A^*_2, \dots, A^*_n)$
    
    \end{algorithmic}
\end{algorithm}

\subsection{Time and Space Complexity Analysis}

\begin{proposition}
    The time and space complexity of the winning allocation deterministic dynamic programming are $O(n3m)$ and $O(n2m)$, respectively.
\end{proposition}

\begin{proof}

Clearly, the initialization loop and the final loop involve $2^m$ and $n$ enumerations, respectively. 
In the dynamic programming loop, for each bidder $i \in [2, n]$ and bundle $|\Ss| = k$, we need to enumerate all $2^k$ subsets of $\Ss$. 
Consequently, the total number of enumerations for bidder $i$ is given by
$
    \sum_{k=0}^m \binom{m}{k} 2^k = (2 + 1)^m = 3^m.
$
Since the computation during each enumeration is $O(1)$, the overall time complexity is $O(2^m) + O((n-1)3^m) + O(n) = O(n3^m)$.

Regarding space complexity, apart from the $O(n2^m)$ needed for input valuations and VVCA parameters, we require an additional $O(n2^m)$ space to store $\MAW$ and $\AB$. Notably, there is no need to store $\MAW'$. Therefore, the total space complexity is $O(n2^m)$.
    
\end{proof}

\section{Proof of \cref{proposition:Z:smooth}}
\propositionZSmooth*

To prove \cref{proposition:Z:smooth}, we first present a useful lemma:

\begin{lemma} \label{lemma:W:differentiable}
    Let $W: \RR^n \to \RR$ be a bounded function such that $\max_{\bmx \in \RR^n} |W(x)| \le M_W$, and denote $\Tilde{W}^\sigma \coloneqq \EE_{\bnoise \sim N(0, 1)^n}[W(\bmx + \sigma\bnoise)]$ with $\sigma > 0$ as the Gaussian smoothing approximation of $W$. Then $\Tilde{W}^\sigma$ is differentiable.  
\end{lemma}

\begin{proof}   

For any $\bmx = (x_1, x_2, \dots, x_n)$, without loss of generality, let's examine the differentiability of $\Tilde{W}^\sigma(\bmx)$ with respect to $x_1$. 
We define $\bmx_{-1} = (x_2, x_3, \dots, x_n)$, and for any $\bnoise \in \RR^n$, we denote $\bnoise_{-1} = (\noise_1, \bnoise_{-1})$.  

\paragraph{Existence of $\frac{\partial \Tilde{W}^\sigma(\bmx)}{\partial x_1}$}
We prove the existence of $\frac{\partial \Tilde{W}^\sigma(\bmx)}{\partial x_1}$ by deriving according to the definition of partial derivative: 
\begin{equation}\label{eq:derivative}
\begin{aligned}
    &\lim_{\Delta x_1\rightarrow0} \frac{\Tilde{W}^{\sigma}(x_1+\Delta x_1,\bmx_{-1})-\Tilde{W}^{\sigma}(x_1,\bmx_{-1})}{\Delta x_1}\\
    &=\lim_{\Delta x_1\rightarrow0} \frac{ \EE_{\bnoise \sim N(0, I_n)}[W(x_1 + \Delta x_1 + \sigma\noise_1, \bmx_{-1} + \sigma\bnoise_{-1})] - \EE_{\bnoise \sim N(0, I_n)}[W(\bmx + \sigma\bnoise)] }{\Delta x_1} \\
    &= \lim_{\Delta x_1\rightarrow0} \frac{ \EE_{\bnoise' \coloneqq (\noise'_1, \bnoise_{-1}) \sim N((\frac{\Delta x_1}{\sigma}, \bm{0}_{n-1}), I_n)}[W(\bmx + \sigma\bnoise') ] - \EE_{\bnoise \sim N(0, I_n)}[W(\bmx + \sigma\bnoise)] }{\Delta x_1}\\
    &=\lim_{\Delta x_1\rightarrow0}  \left(\frac{1}{2\pi}\right)^\frac{n}{2}  \int W(\bm{x}+\sigma \bnoise) \frac{\exp\left( -\frac{1}{2}\|(\noise_1-\frac{\Delta x_1}{\sigma},\bnoise_{-1})\|^2\right)-\exp\left( -\frac{1}{2}\|(\noise_1, \bnoise_{-1})\|^2\right)}{\Delta x_1}d\bnoise \\
    &\overset{(a)}{=} \left(\frac{1}{2\pi}\right)^\frac{n}{2}  \int W(\bm{x}+\sigma \bnoise)\frac{\noise_1}{\sigma}\exp\left(-\frac{1}{2}\|\bnoise\|^2\right) d\bnoise \\
    &= \EE_{\bnoise\sim N(0, I)} \left[ \frac{W(\bmx + \sigma\bnoise)}{\sigma}\noise_1 \right]\\
    &= \EE_{\bnoise\sim N(0, I)} \left[ \frac{W(\bmx + \sigma\bnoise) - W(\bmx)}{\sigma}\noise_1 \right], \quad\text{since } \EE_{\bnoise\sim N(0, I)}[W(\bmx)\noise_1] = 0,
\end{aligned}
\end{equation}
where $(a)$ in \cref{eq:derivative} essentially comes from Taylor expansion. 
To provide its proof in detail, we first define $f(t) \coloneqq \exp\left( -\frac{1}{2}\|(t,\bnoise_{-1})\|^2\right)$. Then we use Taylor expansion of $f(\noise_1-\frac{\Delta x_1}{\sigma})$ at $\noise_1$ with Lagrange remainder, that is,
\begin{equation*}
\begin{aligned}
   &f(\noise_1-\frac{\Delta x_1}{\sigma}) =  f(\noise_1) +  f'(\noise_1)(-\frac{\Delta x_1}{\sigma}) + f''(\xi(\Delta x_1,u
   _1))\frac{(\Delta x_1)^2}{\sigma^2}\\   
   =& f(\noise_1) +  \noise_1\exp\left(-\frac{1}{2}\|(\noise_1, \bnoise_{-1})\|^2\right)\frac{\Delta x_1}{\sigma} \\ &+  (\xi^2(\Delta x_1,u_1)-1)\exp\left(-\frac{1}{2}\|(\xi (\Delta x_1,u_1), \bnoise_{-1})\|^2\right)\frac{(\Delta x_1)^2}{\sigma^2}\\
\end{aligned}
\end{equation*}
where $\xi (\Delta x_1,u_1)$ is between $u_1-\frac{\Delta x_1}{\sigma}$ and $u_1$.  
On top of that, $(a)$ is derived by
\begin{equation}\label{eq:derivative:Taylor}
\begin{aligned}
&\lim_{\Delta x_1\rightarrow0} \int W(\bm{x}+\sigma \bnoise) \frac{\exp\left( -\frac{1}{2}\|(\noise_1-\frac{\Delta x_1}{\sigma},\bnoise_{-1})\|^2\right)-\exp\left( -\frac{1}{2}\|(\noise_1, \bnoise_{-1})\|^2\right)}{\Delta x_1}d\bnoise \\ 
=& \lim_{\Delta x_1\rightarrow0} \int W(\bm{x}+\sigma \bnoise)d\bnoise \\ &\cdot \frac{\noise_1\exp\left(-\frac{1}{2}\|(\noise_1, \bnoise_{-1})\|^2\right)\frac{\Delta x_1}{\sigma} + (\xi^2(\Delta x_1,u_1)-1)\exp\left(-\frac{1}{2}\|(\xi (\Delta x_1,u_1), \bnoise_{-1})\|^2\right)\frac{(\Delta x_1)^2}{\sigma^2}}{\Delta x_1} \\
=& \lim_{\Delta x_1\rightarrow0} \int W(\bm{x}+\sigma \bnoise)d\bnoise \\ &\cdot \left(
        \frac{\noise_1}{\sigma}\exp\left(-\frac{1}{2}\|\bnoise\|^2\right) + (\xi^2(\Delta x_1,u_1)-1)\exp\left(-\frac{1}{2}\|(\xi (\Delta x_1,u_1), \bnoise_{-1})\|^2\right)\frac{\Delta x_1}{\sigma^2} 
    \right) \\
\overset{(b)}{=}& \lim_{\Delta x_1\rightarrow0} \int W(\bm{x}+\sigma \bnoise) \frac{\noise_1}{\sigma}\exp\left(-\frac{1}{2}\|\bnoise\|^2\right)d\bnoise\\
=&\int W(\bm{x}+\sigma \bnoise) \frac{\noise_1}{\sigma}\exp\left(-\frac{1}{2}\|\bnoise\|^2\right)d\bnoise,
\end{aligned}
\end{equation}
where $(b)$ is equivalent to the truth of the following equation:
\begin{equation}\label{eq:derivative:Taylor:(b)}
    \begin{aligned}
        & \lim_{\Delta x_1 \to 0} \left| \Delta x_1\int W(x+\sigma \bnoise)(\xi^2(\Delta x_1,u_1)-1)\exp\left(-\frac{1}{2}\|(\xi (\Delta x_1,u_1), \bnoise_{-1})\|^2\right)d\bnoise \right| \\
        \le& \lim_{\Delta x_1 \to 0} |\Delta x_1|\cdot \Big|W(x+\sigma \bnoise)\Big|\cdot \left|\int (\xi^2(\Delta x_1,u_1)-1)\exp\left(-\frac{1}{2}\|(\xi (\Delta x_1,u_1), \bnoise_{-1})\|^2\right)d\bnoise \right| \\
        \le& \lim_{\Delta x_1 \to 0} |\Delta x_1|\cdot M_W \cdot \left|\int (\xi^2(\Delta x_1,u_1)-1)\exp\left(-\frac{1}{2}\|(\xi (\Delta x_1,u_1), \bnoise_{-1})\|^2\right)d\bnoise \right| \\
        =& \lim_{\Delta x_1 \to 0} |\Delta x_1|\cdot M_W \cdot \left|\int K(\xi(\Delta x_1,u_1),\bnoise_{-1})d\bnoise \right| \\
        \le& \lim_{\Delta x_1 \to 0} |\Delta x_1|\cdot M_W \cdot \int \Big| K(\xi(\Delta x_1,u_1),\bnoise_{-1}) \Big| d\bnoise \\
    \end{aligned}
\end{equation}
where we define $K(s,\bnoise_{-1})\coloneqq(s^2-1)\exp\left(-\frac{1}{2}\|(s, \bnoise_{-1})\|^2\right)$ with $|s-u_1|\leq \left|\frac{\Delta x_1}{\sigma}\right|$, and we prove $(c)$ by showing that when $|\Delta x_1|<\sigma$, the term
\begin{equation}\label{eq:derivative:Taylor:(b):bound}
    \int \Big| K(\xi(\Delta x_1,u_1),\bnoise_{-1}) \Big| d\bnoise \le \int \max_{s \in [u_1 - |\frac{\Delta x_1}{\sigma}|, u_1 + |\frac{\Delta x_1}{\sigma}|]} |K(s,\bnoise_{-1})| d\bnoise
\end{equation}
is bounded: 
Given the definition of $K$, there exists a constant $C_K > 1$ such that when $s>C_K$, $|K(s,\bnoise_{-1})|$ is monotone decreasing with respect to $s$, and when $s<-C_K$, $|K(s,\bnoise_{-1})|$ is monotone increasing with respect to $s$.
Based on that, we discuss the following three cases of the range of $u_1$:
\begin{enumerate}
    \item When $u_1 > 2C_k$, then $s \ge u - \left|\frac{\Delta x_1}{\sigma}\right| > 2C_K - 1 > C_K$, so that $|K(s,\bnoise_{-1})|$ is monotone decreasing with respect to $s$. Therefore we have 
    \begin{equation*}
    \begin{aligned}
         &\int\limits_{u_1>2C_K}\max_{s \in [u_1 - |\frac{\Delta x_1}{\sigma}|, u_1 + |\frac{\Delta x_1}{\sigma}|]} \Big|K(s,\bnoise_{-1})\Big| d\bnoise 
         \le \int\limits_{u_1>2C_K}\left|K(u_1-\left|\frac{\Delta x_1}{\sigma}\right|,\bnoise_{-1})\right| d\bnoise \\
         \le& \int\limits_{u_1>2C_K}\Big|K(u_1-1,\bnoise_{-1})\Big| d\bnoise =: C_1,
    \end{aligned}
    \end{equation*}
    where $C_1$ is a constant independent of $|\Delta x_1|$.
    \item When $-2C_K \le u_1 \le 2C_K$, we have
    \begin{equation*}
    \begin{aligned}
        &\int\limits_{-2C_K\leq u_1\leq 2C_K}\max_{s \in [u_1 - |\frac{\Delta x_1}{\sigma}|, u_1 + |\frac{\Delta x_1}{\sigma}|]} \Big|K(s,\bnoise_{-1})\Big| d\bnoise \\
        \le& \int\limits_{-2C_K\leq u_1\leq 2C_K}\max_{-2C_K\leq u'_1 \leq 2C_K}\left(\max_{s \in [u'_1 - |\frac{\Delta x_1}{\sigma}|, u'_1 + |\frac{\Delta x_1}{\sigma}|]} \Big|K(s,\bnoise_{-1})\Big|\right)d\bnoise \\
        \le& \int\limits_{-2C_K\leq u_1\leq 2C_K}\max_{-(2C_K + 1)\leq s \leq (2C_K + 1)}\Big|K(s,\bnoise_{-1})\Big|d\bnoise =: C_2
    \end{aligned}
    \end{equation*}
    where $C_2$ is a constant independent of $|\Delta x_1|$.
    \item When $u_1 < -2C_K$, then $s \le u + \left|\frac{\Delta x_1}{\sigma}\right| < -2C_K + 1 < C_K$, so that $|K(s,\bnoise_{-1})|$ is monotone increasing with respect to $s$. Similar to case 1,  we have 
    \begin{equation*}
    \begin{aligned}
        &\int\limits_{u_1<-2C_K}\max_{s \in [u_1 - |\frac{\Delta x_1}{\sigma}|, u_1 + |\frac{\Delta x_1}{\sigma}|]} \Big|K(s,\bnoise_{-1})\Big| d\bnoise
        \le \int\limits_{u_1<-2C_K}\left|K(u_1+\left|\frac{\Delta x_1}{\sigma}\right|,\bnoise_{-1})\right| d\bnoise \\
        \le& \int\limits_{u_1<-2C_K}\Big|K(u_1+1,\bnoise_{-1})\Big| d\bnoise =: C_3
    \end{aligned}
    \end{equation*}
    where $C_3$ is a constant independent of $|\Delta x_1|$.
\end{enumerate}

Therefore, \cref{eq:derivative:Taylor:(b):bound} is bounded by $C_1 + C_2 + C_3$ which is independent of $|\Delta x_1|$, and the proof of \cref{eq:derivative:Taylor:(b)} is complete. 
Subsequently, the proof of $(b)$ in \cref{eq:derivative:Taylor} is established. 
Following that, the proof of $(a)$ in \cref{eq:derivative} is concluded. 
Consequently, the existence of $\frac{\partial \Tilde{W}^\sigma(\bmx)}{\partial x_1}$ is derived from \cref{eq:derivative}.

\paragraph{Continuity of $\frac{\partial \Tilde{W}^\sigma(\bmx)}{\partial x_1}$}
Given $\sigma$, let $Q(\bm x) \coloneqq W(\frac{\bm x}{\sigma})$ and $ H(\bnoise)\coloneqq -\noise_1 \exp\left(-\frac{1}{2}\left\|\bnoise\right\|^2\right)$, according to \cref{eq:derivative} we have
\begin{equation*}
    \begin{aligned}
        \frac{\partial \Tilde{W}^\sigma(\bmx)}{\partial x_1}&=\frac{1}{\sigma}\left(\frac{1}{2\pi}\right)^\frac{n}{2}  \int W(\bm{x}+\sigma \bnoise)\noise_1\exp\left(-\frac{1}{2}\|\bnoise\|^2\right) d\bnoise \\
        &=\frac{1}{\sigma}\left(\frac{1}{2\pi}\right)^\frac{n}{2}  \int Q(\bnoise+\frac{\bm x}{\sigma})H(-\bnoise)d\bnoise\\
        &=\frac{1}{\sigma}\left(\frac{1}{2\pi}\right)^\frac{n}{2}  \int Q(\bnoise)H(\frac{\bm x}{\sigma}-\bnoise)d\bnoise\\
        &=: \frac{1}{\sigma}\left(\frac{1}{2\pi}\right)^\frac{n}{2} C(\frac{\bmx}{\sigma}),
    \end{aligned}
\end{equation*}
where we can see that $C(\frac{\bmx}{\sigma})$ is a convolution function. 
As for the continuity of $C(\frac{\bmx}{\sigma})$, we have:
\begin{equation*}
\begin{aligned}
       \lim_{\bm{h}\rightarrow0}\left| C(\frac{\bm{x}}{\sigma}+\bm{h}) - C(\frac{\bm{x}}{\sigma})\right| &=
       \lim_{\bm{h}\rightarrow0}\left|\int H(\frac{\bm{x}}{\sigma}+\bm{h}-\bnoise)Q(\bnoise)d\bnoise -\int H(\frac{\bm{x}}{\sigma}-\bnoise)Q(\bnoise)d\bnoise \right|\\
       &\le \lim_{\bm{h}\rightarrow0}\int \left| H(\frac{\bm{x}}{\sigma}+\bm{h}-\bnoise) - H(\frac{\bm{x}}{\sigma}-\bnoise)\right|\left|Q(\bnoise)\right|d\bnoise\\
       &\overset{(a)}{\le} \lim_{\bm{h}\rightarrow0}M_W\int \left|H(\frac{\bm{x}}{\sigma}+\bm{h}-\bnoise) - H(\frac{\bm{x}}{\sigma}-\bnoise)\right|d\bnoise\\
       &\overset{(b)}{=}0,
\end{aligned}
\end{equation*}
where $(a)$ holds because $Q_{\sigma}(\bm x) \coloneqq W(\frac{\bm x}{\sigma})$ is bounded by $M_W \coloneqq \max_{\bm x} W(\bmx)$, and $(b)$ holds because $ H(\bnoise)$ is integrable. 
As a result, $C(\frac{\bmx}{\sigma})$ is continuous with respect to $\bmx$, and then $\frac{\partial \Tilde{W}^\sigma(\bmx)}{\partial x_1}$ is also continuous with respect to $x_1$.

\paragraph{Differentiability of $\Tilde{W}^\sigma$} 
From the preceding discussion, it is evident that $\Tilde{W}^\sigma(\bmx)$ is differentiable with respect to $x_1$. 
This technique can be extended straightforwardly to $x_2, x_3, \dots, x_n$. 
Therefore, $\Tilde{W}^\sigma(\bmx)$ is differentiable (and thus continuous) with respect to $\bmx$.

\end{proof}

\begin{proof}[Proof of \cref{proposition:Z:smooth}]

$\Tilde{Z}_\Ss^\sigma(e^{\bm\alpha}, \bm\lambda)$ is defined as the Gaussian smoothing approximation of ${Z}_\Ss(e^{\bm\alpha}, \bm\lambda)$, whose range is $[\min_{V \in \Ss}\min_{A \in \Aa}\sum_{i=1}^n v_i(A), \max_{V \in \Ss}\max_{A \in \Aa}\sum_{i=1}^n v_i(A)]$. 
Therefore, the proof is done by applying \cref{lemma:W:differentiable}.
    
\end{proof}

\section{Further Implementation Details}
\label{app:implementation}

We use the same hyperparameters for OD-VVCA and the ablation version FO-VVCA in all the settings, and we list the main hyperparameters in \cref{tab:hyperparams}.

\begin{table*}[h]
    \centering
    \caption{
        Hyperparameters of OD-VVCA and FO-VVCA.
    }
    \vspace{5pt}
    \subtable[Small-scale settings.]{
        \resizebox{\textwidth}{!}{
        \begin{tabular}{lcccccccccc}
        \toprule
        \multirow{2}{*}{Hyperparameter} & \multicolumn{5}{c}{Symmetric} & \multicolumn{5}{c}{Asymmetric} \\
        \cmidrule(l{10pt}r{10pt}){2-6} \cmidrule(l{10pt}r{10pt}){7-11}
        & 2$\times$2\ref{settingA} & 2$\times$5\ref{settingA} & 3$\times$10\ref{settingA} & 2$\times$2\ref{settingD} & 3$\times$10\ref{settingD} &  5$\times$3\ref{settingB} & 3$\times$10\ref{settingB} & 2$\times$5\ref{settingC} & 5$\times$3\ref{settingC} & 3$\times$10\ref{settingC}\\
        \midrule
        Learning Rate & 0.01 & 0.001 & 0.001 & 0.01 & 0.001 & 0.001 & 0.001 & 0.001 & 0.001 & 0.001\\
        $n_r$ & 8 & 8 & 8 & 8 & 8 & 8 & 8 & 8 & 8 & 8 \\
        $\sigma$ & 0.01 & 0.01 & 0.01 & 0.01 & 0.01 & 0.01 & 0.01 & 0.01 & 0.01 & 0.01\\
        Iteration & 2000 & 2000 & 2000 & 2000 & 2000 & 2000 & 2000 & 2000 & 2000 & 2000 \\
        Batch Size & 1024 & 2048 & 1024 & 1024 & 1024 & 1024 & 1024 & 2048 & 1024 & 1024\\
        \bottomrule
        \end{tabular}
        }
    }
    \subtable[Large-scale settings.]{
        \resizebox{\textwidth}{!}{  
        \begin{tabular}{lcccccccccc}
        \toprule
            \multirow{2}{*}{Hyperparameter} & \multicolumn{4}{c}{Symmetric} & \multicolumn{6}{c}{Asymmetric} \\
            \cmidrule(l{10pt}r{10pt}){2-5} \cmidrule(l{10pt}r{10pt}){6-11}
            & 5$\times$10\ref{settingA} & 10$\times$5\ref{settingA} & 5$\times$10\ref{settingD} & 10$\times$5\ref{settingD} & 5$\times$10\ref{settingB} & 10$\times$5\ref{settingB} & 30$\times$5\ref{settingB} & 5$\times$10\ref{settingC} & 10$\times$5\ref{settingC} & 30$\times$5\ref{settingC} \\
        \midrule
        Learning Rate & 0.0003 & 0.0003 & 0.0003 & 0.0003 & 0.005 & 0.005 & 0.005 & 0.005 & 0.005 & 0.005\\
        $n_r$ & 8 & 8 & 8 &8 & 8 & 8 & 8 & 8 & 8 & 8 \\
        $\sigma$ &0.001 & 0.01 & 0.01 & 0.01 & 0.01 & 0.01 & 0.01 & 0.01 & 0.01 & 0.01\\
        Iteration & 2000 & 2000 & 2000 & 2000 & 2000 & 2000 & 2000 & 2000 & 2000 & 2000 \\
        Batch Size & 1024 & 1024 & 1024 & 1024 & 1024 & 1024& 1024 & 1024 & 1024 & 1024 \\
        \bottomrule
        \end{tabular}
        }
    }
    \label{tab:hyperparams}
\end{table*}

For Lottery AMA~\citep{curry2022differentiable} and AMenuNet~\citep{duan2023scalable}, we adopt the same hyperparameters as those specified in their original publications, with the exception of the candidate size. 
Due to computational resource limitations, the size of the candidate allocation set cannot be excessively large. 
Therefore, we set the candidate size based on the complexity of the auction scenario to make a balance between performance and computational feasibility, as detailed in \cref{tab:menusize}.

\begin{table*}[h]
    \centering
    \caption{
        Candidiate size of Lottery AMA and AMenuNet in different settings.
    }
    \vspace{5pt}
    \subtable[Small-scale settings.]{
        \resizebox{\textwidth}{!}{
        \begin{tabular}{lcccccccccc}
        \toprule
        \multirow{2}{*}{Method} & \multicolumn{5}{c}{Symmetric} & \multicolumn{5}{c}{Asymmetric} \\
        \cmidrule(l{10pt}r{10pt}){2-6} \cmidrule(l{10pt}r{10pt}){7-11}
        & 2$\times$2\ref{settingA} & 2$\times$5\ref{settingA} & 3$\times$10\ref{settingA} & 2$\times$2\ref{settingD} & 3$\times$10\ref{settingD} &  5$\times$3\ref{settingB} & 3$\times$10\ref{settingB} & 2$\times$5\ref{settingC} & 5$\times$3\ref{settingC} & 3$\times$10\ref{settingC} \\
        \midrule
        Lottery AMA & 32 & 128 & 1024 & 32 & 1024 & 128 & 1024 & 128 & 128 & 1024\\
        AMenuNet & 32 & 128 & 1024 & 32 & 1024 & 128 & 1024 & 128 & 128 & 1024\\
        \bottomrule
        \end{tabular}
        }
    }
    \subtable[Large-scale settings.]{
        \resizebox{\textwidth}{!}{  
        \begin{tabular}{lcccccccccc}
        \toprule
            \multirow{2}{*}{Method} & \multicolumn{4}{c}{Symmetric} & \multicolumn{6}{c}{Asymmetric} \\
            \cmidrule(l{10pt}r{10pt}){2-5} \cmidrule(l{10pt}r{10pt}){6-11}
            & 5$\times$10\ref{settingA} & 10$\times$5\ref{settingA} & 5$\times$10\ref{settingD} & 10$\times$5\ref{settingD} & 5$\times$10\ref{settingB} & 10$\times$5\ref{settingB} & 30$\times$5\ref{settingB} & 5$\times$10\ref{settingC} & 10$\times$5\ref{settingC} & 30$\times$5\ref{settingC} \\
        \midrule
        Lottery AMA & 4096 & 4096& 4096& 4096& 4096& 4096& 4096& 4096& 4096& 4096\\
        AMenuNet & 4096& 4096& 4096& 4096& 4096& 4096& 4096& 4096& 4096& 4096\\
        \bottomrule
        \end{tabular}
        }
    }
    \label{tab:menusize}
\end{table*}

\section{Training Time} \label{app:time}

We document the total training time of Lottery AMA~\citep{curry2022differentiable}, AMenuNet~\citep{duan2023scalable}, BBBVVCA~\citep{VVCA2015}, and the proposed OD-VVCA for different auction scales in \cref{tab:time}.
We tuned the number of training iterations for each method based on our empirical observations of their convergence rates, which are also listed in the table.

\begin{table*}[h]
    \centering
    \caption{
        The total training time for different methods.
        We use the same parallelizable dynamic programming for winning allocation determination in both BBBVVCA and OD-VVCA.
    }
    \vspace{10pt}
    \resizebox{\textwidth}{!}{  
    \begin{tabular}{lccccccc}
    \toprule
       Training Time & $2\times 2$ & $2\times 5$ & $5\times 3$ & $10\times 5$ & $30\times 5$ & $3\times 10$ & $5\times 10$  \\
    \midrule
     Lottery AMA (10000 iterations) &     1min40s       &       15min      &    8min        &     20min       & 6h  &         1h40min    &       9h     \\
     AMenuNet (1000 iterations) & 12min      & 22min       & 24min      & 1h20min   & {50min}*   & 1h20min     & 3h30min     \\
     BBBVVCA  (4000 iterations) & 1min      & 3min         & 3min        & 16min     & 1h    & 2h       & 4h20min      \\
     OD-VVCA (2000 iterations) & 1min20s      & 7min        & 6min       & 40min    & 3h    & 6h          & 9h     \\
    \bottomrule
    \end{tabular}
    }
    \par $^*$We reduce the batch size of AMenuNet in the $30 \times 5$ scenarios due to GPU constraints. 
    \label{tab:time}
\end{table*}

As indicated in the table, OD-VVCA demonstrates acceptable training times across all settings. 
This result, along with the similarly acceptable training times of BBBVVCA, underscores the efficiency of the parallelized dynamic programming winner determination algorithm. 
Furthermore, the relatively small number of iterations highlights the effective convergence of our objective decomposition method.

\end{document}